\documentclass[a4paper,11pt]{article}

\usepackage{amsthm,amsmath,amssymb,amsfonts,pstricks,pst-node,graphicx}

\DeclareMathAccent{\wtilde}{\mathord}{largesymbols}{"65}
\DeclareMathAccent{\what}{\mathord}{largesymbols}{"62}

\newcommand\C{{\mathbb C}}
\newcommand\bbbr{{\mathbb R}}

\newcommand\cL{{\mathcal L}}

\newcommand\cA{{\mathcal A}}
\def\bbbc{{\mathbb C}}
\def\bbbr{{\mathbb R}}

\renewcommand\i{{\iota}}
\def\im{\operatorname{Im}}
\def\re{\operatorname{Re}}

\newcommand\ba{{\bf a}}

\newtheorem{Rem}{Remark}

\newtheorem{Pro}{Proposition}

\makeatletter
\def\wb{\accentset{{\cc@style\underline{\mskip10mu}}}}
\makeatother

\begin{document}
\title{Dressing method for the vector sine-Gordon equation and its
soliton interactions}
\author{ Alexander V. Mikhailov$^{\star}$, Georgios Papamikos$^{\dagger}$ and 
Jing Ping Wang$ ^\dagger $
\footnote{Corresponding author: E-mail address: j.wang@kent.ac.uk and Tel: +44 
1227 827181}
\\
$\dagger$ School of Mathematics, Statistics \& Actuarial Science, University of Kent, Canterbury, UK \\
$\star$ School of Mathematics, University of Leeds, Leeds, UK
}
\date{}
\maketitle
\begin{abstract}
In this paper, we develop the dressing method to study the exact solutions for 
the vector sine-Gordon equation.
The explicit formulas for one kink and one breather are derived. The method can be used to construct multi-soliton solutions. 
Two soliton interactions are also studied. The formulas for position shift of 
the kink and position and phase shifts of the breather 
are given. These quantities only depend on the pole positions of the dressing
matrices. 
\end{abstract}

\section{Introduction}
This paper is devoted to the study of an $O(n)$-invariant generalisation of the
sine-Gordon equation
\begin{equation}\label{vsG}
 D_t \left(\frac{{\vec \alpha}_x}{\beta}\right)={\vec \alpha}, \quad 
\beta^2+\langle{\vec \alpha}, {\vec \alpha}\rangle=1,
\end{equation}
where the dependent variable ${\vec \alpha}=(\alpha^1, \cdots, \alpha^n)^T$ is 
$n$-dimensional 
real vector field and $\beta\in\bbbr$. Here and in what follows the upper index
$T$ denotes the transpose of a vector or a matrix.
We use the notation  $\langle\cdot, \cdot\rangle$ for the Euclidean dot product
of two vectors.

Equation (\ref{vsG}) was first appeared in \cite{PoR79} viewed as a reduction of
the two-dimensional $O(n)$ nonlinear $\sigma$-model \cite{mr80c:81115}. Its
integrability properties were further studied afterwards. The Lax pairs was
given in \cite{EP79} and its Lagrangian formulation in \cite{bps95}. Later,
this equation reappeared in the study of connection between finite dimensional
geometry, infinite dimensional geometry and integrable systems \cite{wang02}.
It was derived as the inverse flow of the vector modified Korteweg-de Vries
equation
\begin{eqnarray}\label{vKdV}
{\vec u}_{\tau}={\vec u}_{xxx}+\frac{3}{2} \langle{\vec u}, {\vec u}\rangle 
{\vec u}_{x},
\quad {\vec u}= \frac{{\vec \alpha}_x}{\beta},
\end{eqnarray}
whose Hamiltonian, symplectic and
hereditary recursion operators were naturally derived using the structure
equation for the evolution of a curve embedded in an $n$-dimensional Riemannian
manifold with constant curvature \cite{MR2058803}. These have been recently 
re-derived in \cite{anco06}.  Besides, a partial classification of vector
sine-Gordon equations using symmetry tests was
done in \cite{ancow05}.

Equation (\ref{vsG}) is a higher-dimensional generalisation of the well-known scalar sine-Gordon equation
\begin{equation}\label{ssg}
 \theta_{xt}=\sin \theta .
\end{equation}
Indeed, it can be obtained by taking the dimension $n=1$ and letting $\beta=\cos \theta$ and $\alpha^1=\sin \theta$.
The scalar sine-Gordon equation originates in differential geometry and has profound applications in physics
and in life sciences (see recent review \cite{II13}). Vector generalisations of 
integrable equations have proved to be useful in applications \cite{manakov}. 
They can be associated with symmetric spaces \cite{FK}.

The rational dressing method was originally proposed in \cite{mr80c:81115, zash}
and 
developed in  \cite{mik81}. This method enables one to construct multi-soliton 
solutions and analyse soliton interactions in detail using basic knowledge of Linear 
Algebra. In this paper, we develop the dressing method for the vector sine-Gordon 
equation (\ref{vsG}) and show that similar to the scalar sine-Gordon 
equation (\ref{ssg}) there are two distinct types of solitons, namely kinks and 
breathers. One kink solution is a stationary wave propagating with a constant 
velocity. We show that a kink solution of the vector sine-Gordon equation can be 
obtained from a kink  solution of (\ref{ssg}) by setting $\vec{\alpha}={\bf 
a}\sin\theta,\ \beta=\cos \theta$, where ${\bf a}$ is a constant unit length 
vector in $\bbbr^n$. A general two kink solution of (\ref{vsG}) cannot be 
obtained from solutions of (\ref{ssg}), but it can be seen as a two kink 
solution of a vector sine-Gordon equation (\ref{vsG}) with $n=2$. One breather 
solution is a localised and periodically oscillating wave moving with a constant 
velocity. One breather solution of the general $O(n)$ invariant equation 
(\ref{vsG}) can 
be obtained from a breather solution of (\ref{vsG}) with $n=2$ by an 
appropriate $O(n)$ rotation. Two breathers solution can be obtained from the 
corresponding solution of (\ref{vsG}) with $n=4$, etc. Surprisingly, 
the effects of interaction, such as the displacement and a phase shift (for 
breathers) are exactly the same as in the case of the scalar sine-Gordon 
equation (\ref{ssg}) \cite{fad87}. Such interaction properties are naturally valid for
the vector modified Korteweg-de Vries equation (\ref{vKdV}). The detailed study of soliton interactions 
for (\ref{vKdV}) when $n=2$ can be found in \cite{anco11}.

\section{Dressing method for the vector sine-Gordon equation}\label{sec2}
In this section, we begin with the Lax representation of the vector sine-Gordon equation (\ref{vsG}) given in \cite{wang02},
which is invariant under the reduction group $\mathbb{Z}_2\times \mathbb{Z}_2\times
\mathbb{Z}_2$.
We then study the conditions for the dressing matrix (assumed to be rational in spectral parameter) with the same symmetries.
The $1$-soliton solutions of (\ref{vsG}) correspond the dressing matrix with only simple poles belonging to a single orbit of the reduction 
group. For one kink, it has two pure imaginary simple poles and for one 
breather, it has four complex simple poles.  
Using the dressing method, we explicitly derive one kink and one breather
solutions starting with a trivial solution.

The vector sine-Gordon equation (\ref{vsG}) is equivalent to the compatibility 
condition \cite{wang02}
$
 [\cL,\cA]=0
$
for two linear problems
\begin{equation}\label{linprobs}
  \cL \Psi=0,\qquad \cA \Psi=0, 
\end{equation}
where
\begin{eqnarray}\label{lax}
 \cL=D_x-\lambda J -U\quad \mbox{and} \quad \cA=D_t+\lambda^{-1} V,
\end{eqnarray}
and 
\begin{eqnarray}\label{JUV}
J=\left(\begin{array}{ccc} 0 & 1 &{\bf 0}^T\\-1 & 0 & {\bf 0}^T\\{\bf 0} &{\bf 0}&0_n \end{array}\right),\
U=\left(\begin{array}{ccc} 0 & 0 &{\bf 0}^T\\0 & 0 & -{\vec \alpha}_x^T/\beta\\{\bf 0} &
{\vec \alpha}_x/\beta&0_n \end{array}\right),\ 
V=\left(\begin{array}{ccc} 0 & \beta &{\vec \alpha}^T\\-\beta & 0 & 
{\bf 0}^T\\-{\vec \alpha}
&{\bf 0}&0_n \end{array}\right)\ ,
\end{eqnarray}
where ${\bf 0}$ is $n$-dimensional zero column vector and $0_n$ is the $n\times
n$ zero matrix.
Without causing confusion, we sometimes simply write $0$ instead.

The Lax operators $\mathcal{L}$ and $\mathcal{A}$ are invariant under the 
(reduction) group
of automorphisms generated by the following three
transformations: the first one is
\begin{equation}
\i:\mathcal{L}(\lambda)\rightarrow -\mathcal{L}^{\dagger}(\lambda),
\label{redgroup1}
\end{equation}
where $\mathcal{L}^{\dagger}(\lambda)$ is the
adjoint operator defined by $\mathcal{L}^{\dagger}(\lambda)=-D_x-\lambda
J^T-U^T$. The invariance under this transformation implies the matrices $J$ and
$U$ are skew-symmetric. The second one is
\begin{equation}
r:\mathcal{L}(\lambda)\rightarrow \overline{ \mathcal{L}(\bar{\lambda})},
\label{redgroup3}
\end{equation}
where $\overline{\mathcal{L}(\lambda)}$ means its complex conjugate.
The invariance under this transformation reflects that the entries of matrices
$U$ and $V$ are
real. The last one is called Cartan involution
\begin{equation}
s:\mathcal{L}(\lambda)\rightarrow Q\mathcal{L}(-\lambda)Q,
\label{redgroup2}
\end{equation}
where $Q=\mbox{diag}(-1,1,\ldots,1)$, which leads to the reduction to the
symmetric space.

These three commuting transformations satisfy
\begin{eqnarray*}
 \i^2=r^2=s^2={\rm id} 
\end{eqnarray*}
and therefore generate the group $\mathbb{Z}_2\times \mathbb{Z}_2\times
\mathbb{Z}_2$. Indeed, the operator $\cA$ is also invariant under it, that is,
\begin{equation}
\i(\cA(\lambda))=\cA(\lambda), \qquad r(\cA(\lambda))=\cA(\lambda),\quad
 s(\cA(\lambda))=\cA(\lambda).
\label{eq:symconstr1}
\end{equation}
Thus we say the Lax representation of (\ref{vsG}) is invariant under the reduction group \cite{mik81,mik79,mik80} $\mathbb{Z}_2\times \mathbb{Z}_2\times
\mathbb{Z}_2$.

In what follows, we use the method of rational dressing \cite{mr80c:81115, 
zash, mik81} to 
construct 
new exact solutions of (\ref{vsG}) 
starting from an exact solution ${\vec \alpha_0},\ \beta_0$. Let us denote by
$U_0, V_0$ the matrices $U,V$ in which ${\vec \alpha},\ \beta$ are 
replaced by the exact solution ${\vec \alpha_0},\ \beta_0$ of (\ref{vsG}). The 
corresponding overdetermined linear system  
\begin{equation}
\label{naked}
\cL_0 \Psi_0=(D_x-\lambda J-U_0) \Psi_0=0 \qquad 
\cA_0 \Psi_0=(D_t+\lambda^{-1} V_0)\Psi_0=0 
\end{equation}
has a fundamental solution $\Psi_0(\lambda,x,t)$ invariant under transformations
(\ref{redgroup1})--(\ref{redgroup2}). Following
\cite{mr80c:81115, zash} we shall 
assume that the fundamental solution $\Psi(\lambda,x,t)$ for the new 
(``dressed'') linear problems 
\begin{equation}
\label{dressed}
\cL \Psi=(D_x-\lambda J-U) \Psi=0 \qquad 
\cA \Psi=(D_t+\lambda^{-1} V)\Psi=0 
\end{equation}
is of the form 
\begin{eqnarray}\label{dress}
 \Psi=\Phi(\lambda) \Psi_0, \quad \det \Phi\neq 0,
\end{eqnarray}
where the dressing matrix $\Phi(\lambda)$ is assumed to be rational in the 
spectral parameter $\lambda$ and to be invariant with respect to symmetries
\begin{eqnarray}
&& \Phi(\lambda)^{-1}=\Phi^{T}(\lambda), \qquad
\overline{\Phi(\bar{\lambda})}=\Phi(\lambda), \qquad Q \Phi(-\lambda)
Q=\Phi(\lambda).
\label{Inphi}
\end{eqnarray}
Conditions (\ref{Inphi}) guarantee that the corresponding Lax operators 
$\mathcal{L}$ and $\mathcal{A}$ are 
invariant with respect to transformations
(\ref{redgroup1})--(\ref{redgroup2}).

It follows from (\ref{naked}), (\ref{dressed}) and (\ref{dress}) that 
\begin{eqnarray}
&& \Phi (D_x -\lambda J-U_0) \Phi^{-1}=-\lambda J-U; \label{chix}\\
&& \Phi (D_t+\lambda^{-1} V_0) \Phi^{-1}=\lambda^{-1} V. \label{chit}
\end{eqnarray}
These equations enable us to specify the form of the dressing matrix $\Phi$ and 
construct the corresponding ``dressed'' solution $\vec{\alpha},\beta$ of the 
vector sine-Gordon equation (\ref{vsG}).

Let us consider the most
trivial case when the dressing matrix $\Phi$ does not 
depend on the spectral parameter $\lambda$. In this case the dressing results 
in a point transformation ($O(n)$ rotation) of the initial solution  
$\vec{\alpha}_0$.
\begin{Pro}
Assume that $\Phi$ is a $\lambda$ independent dressing matrix for the vector sine-Gordon equation (\ref{vsG}). If it is
invariant with respect to symmetries (\ref{Inphi}),  then 
\begin{eqnarray}\label{phinola}
 \Phi=\pm \left(\begin{array}{ccc}
                 1&0&{\bf 0}^T\\
                 0&1&{\bf 0}^T\\
                 {\bf 0}&{\bf 0}&\Omega
                \end{array}\right),
\end{eqnarray}
where $\Omega\in O(n,\bbbr)$ is a constant ($x,t$-- 
independent) matrix. The corresponding ``dressed'' 
solution is $\vec{\alpha}=\Omega\, \vec{\alpha}_0,\ \beta=\beta_0$.
\end{Pro}
\begin{proof} Under the assumption that the dressing matrix $\Phi$ is
independent of the spectral
parameter $\lambda$, it follows from (\ref{chix}) and (\ref{chit}) that
\begin{eqnarray}\label{consd}
[J,\Phi]=0; \quad \Phi_x=U \Phi-\Phi U_0; \quad \Phi_t=0; \quad V \Phi-\Phi V_0=0.
\end{eqnarray}
It follows from the first condition $[J, \Phi]=0$ that the matrix $\Phi$ is of
the form
\begin{eqnarray*}
 \Phi=\left(\begin{array}{ccc} m_{11} & m_{12} & {\bf 0}^T\\-m_{12} &m_{11} &{\bf 0}^T\\ {\bf 0}& {\bf 0} &\Omega\end{array} \right),
\end{eqnarray*}
where $\Omega$ is an $n\times n$ matrix. Substituting it into the last condition 
 in (\ref{consd}), we get
$m_{12}=0$ and
$$ ({\beta}-\beta_0)m_{11}=0; \quad {\vec \alpha}^T \Omega=m_{11} {\vec \alpha_0}^T;\quad 
m_{11}{\vec \alpha} =\Omega {\vec \alpha_0}.
$$
This implies that $\beta=\beta_0$ since $\det \Phi\neq 0$ and $\Omega \Omega^T=m_{11}^2 I_n$,
where $I_n$ is the $n\times n$ identity matrix.
The rest two conditions in (\ref{consd}) imply that matrix $\Phi$ is independent
of $x$ and $t$. 
Finally, since $\Phi$ satisfies (\ref{Inphi}), we have that matrix $\Omega$ is 
real and $m_{11}^2=1$, i.e. $m_{11}=\pm 1$.
\end{proof}

A $\lambda$--dependent dressing matrix $\Phi(\lambda)$, which is
invariant with respect to symmetries (\ref{Inphi}) has poles at the orbits of 
the reduction group.
Simplest ``one soliton'' dressing correspond to the cases when  matrix
$\Phi(\lambda)$ has only simple poles belonging to a single 
orbit. There are two non-trivial cases:
\begin{enumerate}
 \item[(i)] Matrix $\Phi(\lambda)$ has two pure imaginary poles at $\pm i\nu,\ 
\nu\ne 0,\ \nu\in\bbbr$. 
\item[(ii)] Matrix $\Phi(\lambda)$ has four complex poles at $\pm \mu,\ 
\pm\bar{\mu},\ 
\mu=\gamma+i\omega,\  \gamma\omega \ne 0,\ \gamma,\omega\in\bbbr$.
\end{enumerate}
These cases correspond to the ``kink'' and ``breather'' solutions respectively. 
We will show that the case of  two real poles at $\pm\mu\in\bbbr$ leads to a 
trivial dressing.

A more general case in which the poles belong to a finite union of orbits
corresponds to a multi-soliton solution. Using this approach we can construct
explicitly multi-kink-breather solutions 
and analyse the result of kink-kink, breather-breather and kink-breather 
collisions.  

The explicit forms of matrix $\Phi(\lambda)$ corresponding the above two cases
and satisfying the last two conditions in (\ref{Inphi}) are 
\begin{eqnarray*}
({\rm i})&\quad & \Phi(\lambda)=C+\frac{A}{\lambda- \mu}-\frac{QAQ}{\lambda+
\mu}, \quad A=-Q \bar{A} Q, \quad \mu=i\nu; \\
({\rm ii})&\quad&\Phi(\lambda)=C+\frac{A}{\lambda-\mu}-\frac{QAQ}{\lambda+\mu}
+\frac { \bar A}{
\lambda-{\bar \mu} } -\frac{Q{\bar A}Q}{\lambda+{\bar \mu}} ,\quad
\mu=\gamma+i\omega,
\end{eqnarray*}
where $C=QCQ$ is a real matrix of size $(n+2)\times(n+2)$ independent of $\lambda$. 

We now derive the conditions on matrices $A$ and $C$ such that $\Phi(\lambda)$ 
is an invariant dressing matrix, i.e. satisfying conditions (\ref{Inphi}).
First we have
\begin{Pro}\label{prom}
 If $\Phi(\lambda)$ is a dressing matrix satisfying $\Phi(\lambda)\Phi^T(\lambda)=I_{n+2}$, then matrix $C$
 is a constant matrix of the same form as (\ref{phinola}).
\end{Pro}
\begin{proof} All identities must be satisfied for all values of $\lambda$.
For (\ref{chix}), the linear terms in $\lambda$ as $\lambda 
\rightarrow \infty$ leads to $[J, C]=0$. So matrix $C$ is of the form
\begin{eqnarray*}
 C=\left(\begin{array}{ccc} c_{11} & c_{12} & {\bf 0}^T\\-c_{12} &c_{11} &{\bf 0}^T\\ {\bf 0}& {\bf 0} &\Omega\end{array} \right),
\end{eqnarray*}
where $\Omega$ is an $n\times n$ matrix. Substituting it into the last condition $C=QCQ$, we get
$c_{12}=0$.

The constant terms of (\ref{chix}) in $\lambda$ as $\lambda \rightarrow \infty$ leads to
\begin{eqnarray}\label{recu0}
 C_x=UC-CU_0-[B, J],
\end{eqnarray}
where $B$ is the coefficient at $\lambda^{-1}$ in the expansion 
$
\Phi(\lambda)=C+B \lambda^{-1}+\cdots
$
as $\lambda\rightarrow \infty$. Substituting the form of $C$ into it, we get $C_x=0$ implying that matrix $C$ is independent of $x$.
For (\ref{chit}) and $\Phi(\lambda)\Phi^T(\lambda)=I_{n+2}$, taking the limit $\lambda\rightarrow \infty$, we get $C_t=0$
and $ C C^T=I_{n+2}$. Thus we proved the statement.
\end{proof}
Note that a composition of two dressing matrices is also a dressing matrix.
Therefore, without loss of generality, from now on, we take $C=I_{n+2}$ in the dressing matrix $\Phi(\lambda)$.

\begin{Pro}\label{pro0}
The dressing matrix satisfies $\Phi(\lambda)\Phi^T(\lambda)=I_{n+2}$ if and only
if
\begin{eqnarray}\label{Akb}
 AA^T=0; \qquad  
 P A^T
+A P^T=0, 
\end{eqnarray}
where
\begin{eqnarray}\label{defp}
 P=\left(\Phi(\lambda)-\frac{A}{\lambda-\mu}\right){\big
|}_{\lambda=\mu}. 
\end{eqnarray}
\end{Pro}
\begin{proof}
The product  $\Phi(\lambda)\Phi^T(\lambda)$ is a rational
matrix function of $\lambda$ and at $\lambda =\infty$ it is equal to the unit matrix. Now
it is sufficient to show that the rest of conditions in (\ref{Akb}) are
equivalent to vanishing all other poles of the product. The first condition of
(\ref{Akb}) is
equivalent to vanishing of the second order pole at $\lambda=\mu$. All other
second order poles of the product vanishes due to the reduction group symmetry.
The residue at $\lambda=\mu$ vanishes if and only if $A$ satisfies the second
condition of (\ref{Akb}).
\end{proof}
Therefore,  the matrix $A$ is degenerate. Here we shall study solution of rank
one, that is,  $\mbox{rank} (A)=1$ and thus represent the matrix $A$ by a bi-vector
\begin{equation}
  A=b\rangle\langle a,\qquad \langle a=(p,q,a^1,\ldots, a^{n}),\qquad  b\rangle
=(b^1,\ldots , b^{n+2})^T\, , \label{bi-vector}
\end{equation}
where $p, q, a^i, b^i\in\bbbc$. We use different notations for the first two 
component of the vector  $\langle a$ in order to emphasise that they play a 
particular role in the solutions. 

Under the assumption that $A$ is a bi-vector,  we rewrite conditions
(\ref{Akb}) for $A$ into conditions for vectors $a\rangle$
and $b\rangle$.
\begin{Pro}\label{pro00}
Let $A=b\rangle\langle a\ne 0$ be a bi-vector. The dressing matrix satisfies
$\Phi(\lambda)\Phi^T(\lambda)=I_{n+2}$ if and only if
\begin{eqnarray}\label{AkbP}
 \langle a a\rangle=0; \qquad  
 P a\rangle=0,
\end{eqnarray}
where $P$ is given by (\ref{defp}).
\end{Pro}
\begin{proof}
The first equation of (\ref{AkbP}) immediately follows from the first equation of
(\ref{Akb}).
For matrix $A=b\rangle\langle a$ the second condition of (\ref{Akb}) is
equivalent to
\begin{eqnarray*}
  P A^T=0.
\end{eqnarray*}
Indeed, the first term of the sum in the second equation of (\ref{Akb}) is a
bi-vector with image  and  co-image spaces spanned by  $b\rangle$ and  $\langle
b$ respectively and thus it is equal to $\delta b\rangle\langle b$ , that is, 
\[
 PA^T=\delta b\rangle\langle b
\]
for some
$\delta\in\bbbc$. The second term of the sum in the second equation of (\ref{Akb}) is just the
transposition of the first one and thus we get $$P A^T
+A P^T=2\delta b\rangle\langle b=0.$$
Thus  $\delta=0$, that is,  $P A^T=0$ implying   $P a\rangle=0$.
\end{proof}

We now investigate the conditions (\ref{chix}) and (\ref{chit}) that
$\Phi(\lambda)$ must satisfy in order to be a dressing matrix.
Notice that they are rational matrix functions
of $\lambda$. We compare the residues of all poles of both sides.
For (\ref{chix}),  we know that the linear terms in $\lambda$ vanishes as
$\lambda 
\rightarrow \infty$ from Proposition \ref{prom} and 
its constant terms in $\lambda$ as $\lambda \rightarrow \infty$ (\ref{recu0}) becomes 
\begin{eqnarray}\label{recu}
 U=U_0+[B, J] ,
\end{eqnarray}
where $B$ is defined by $\Phi(\lambda)=I_{n+2}+B \lambda^{-1}+\cdots$
as $\lambda\rightarrow \infty$.
We now compute residue of (\ref{chix}) at
$\lambda=\mu$. This leads to
\begin{eqnarray} \label{resi}
 -AJ A^T+A (D_x-\mu
J-U_0)P^T+ P (D_x-\mu 
J-U_0) A^T=0,
\end{eqnarray}
where $A=b\rangle\langle a\ne 0$ satisfying (\ref{AkbP}). In (\ref{resi}) the first term
vanishes since $A$ is a bi-vector and $J$ is a skew-symmetric matrix.  
It follows from vanishing of the second order pole at $\lambda=\mu$ that
\begin{eqnarray}\label{dou}
 A (D_x-\mu J-U_0) A^T=0.
\end{eqnarray}
Thus we can see that 
\begin{eqnarray}\label{ax}
 \cL_0 a\rangle= (D_x-\mu J-U_0) a\rangle=0
\end{eqnarray}
is a solution of (\ref{resi}) and (\ref{dou}) due to (\ref{AkbP}). In fact, the general solution
should be $ \cL_0 a\rangle=\gamma a\rangle$, where $\gamma$ is an
arbitrary scalar function of $x$ and $t$. Without any loss of generality we
set $\gamma=0$. Indeed, the $\gamma$  can be removed by the rescaling $a\rangle$
to $\exp(\int \gamma {\rm d}x) a\rangle$ and it will be shown that
the matrix $A$ is not affected by the scaling of $a\rangle$.

We carry out the similar 
analysis for (\ref{chit}) and obtain that
\begin{eqnarray}
&& V =\Phi(0) V_0 \Phi^T(0); 
\quad (\mbox{residue at $\lambda=0$})\label{recv}\\
&& A(D_t+\mu^{-1} V_0) A^T=0; \quad (\mbox{double pole at 
$\lambda=\mu$})\label{dout}\\
&& -\mu^{-2}A V_0 A^T+A (D_t+\mu^{-1} V_0) P^T+P (D_t+\mu^{-1} V_0) A^T=0.
\quad (\mbox{residue at $\lambda=\mu$})\label{resit} 
\end{eqnarray}
 For $A=b\rangle\langle a$, it is obviously $A V_0 
A^T=0$ since $V_0$ is a skew-symmetric matrix.
In the same way as we get (\ref{ax}) we can see that
\begin{eqnarray}\label{at}
 (D_t+\mu^{-1}V_0) a\rangle=0
\end{eqnarray}
is a solution of (\ref{resit}) and (\ref{dout}).
Thus for given constant vector  $a_1\rangle$ satisfying $\langle a_1 a_1 
\rangle=0$, we obtain that
\begin{eqnarray}\label{veca}
 a\rangle=\Psi_0(x,t,\mu) a_1\rangle .
\end{eqnarray}
To construct the exact solutions for the vector sine-Gordon equation
(\ref{vsG}) using (\ref{recv}) or (\ref{recu}), it is required to determine
the vector $b\rangle$ in the bi-vector $A$ using (\ref{AkbP}). The latter depends
on the choice of the forms for $\Phi(\lambda)$. We will determine the vector
$b\rangle$ for the kink and breather solutions in the following sections.

In the following two sections \ref{Kink} and \ref{sec_br}, we construct the exact solutions starting with the trivial solution $\beta_0=1$
and ${\vec\alpha}_0=0$ for the equation (\ref{vsG}). Then $U_0=0$ and $V_0=J$.
It is easy to see that in this case the fundamental solution for (\ref{naked})
is
\begin{eqnarray}\label{psi0}
\Psi_0(x,t,\lambda)=\exp(\lambda J x-\lambda^{-1} J t)=\left(\begin{array}{ccc} 
\cos \theta&
\sin \theta&0\\ -\sin \theta& \cos \theta&0\\0&0&I\end{array}\right),\quad
\theta=\lambda x-\lambda^{-1} t .
\end{eqnarray}
Matrix $\Psi_0(x,t,\lambda)$ obviously 
satisfies the reduction group symmetry conditions  (\ref{Inphi}).

\subsection{A kink solution}\label{Kink}

Let us assume that the dressing matrix is regular at $\lambda=\infty$ and has 
only two simple poles at points $\lambda=\pm\mu\neq 0$.  The set of poles
is an orbit of the reduction group if $\mu$ is either real or pure
imaginary. Then, without any loss of 
generality, it can be written in the form 
\begin{equation} 
 \Phi(\lambda)=I_{n+2}+\frac{A}{\lambda-\mu}-\frac{QAQ}{\lambda+\mu} .\label{chik}
\end{equation}
For the  bi-vector $A=b\rangle\langle a\ne 0$, we now determine the vector
$b\rangle$
using Proposition \ref{pro00}.

\begin{Pro}\label{pro2}
Let $A=b\rangle\langle a\ne 0$ be a bi-vector . The dressing matrix (\ref{chik})
satisfies $\Phi(\lambda)\Phi^T(\lambda)=I_{n+2}$ if and only if
\begin{eqnarray}\label{abk}
 \langle a a\rangle=0; \qquad
 b\rangle=\frac{2 \mu Q a\rangle}{\langle a Q a \rangle} .
\end{eqnarray}
\end{Pro}
\begin{proof}
 The first equation of (\ref{abk}) has been proved in Proposition \ref{pro00}.
For matrix $A=b\rangle\langle a$ the second condition of (\ref{AkbP}) is
$$
 (I_{n+2}-\frac{Q b\rangle\langle a Q}{2\mu}) a \rangle=0, \quad \mbox{that is},
\quad a \rangle=\frac{Q b\rangle\langle a Q  a \rangle}{2\mu}.
$$
This leads to the second equation of (\ref{abk}) since $Q^2=I_{n+2}$.
\end{proof}
\begin{Rem}\label{rem0}
 It follows from this proposition that the bi-vector $A=b\rangle\langle a$ is parametrised by a subspace
spanned by vector $a\rangle$ rather than vector $a\rangle$ itself. Indeed,
scaling $a\rangle \mapsto \hat \gamma  a\rangle$ does not change
matrix $A$.
\end{Rem}

\begin{Rem}\label{rem1}
As we mention above, there are two possibilities for a two points orbit, namely 
$\mu=\pm \nu$ and $ \mu=\pm i\nu, \ \nu\in \bbbr$. In the first 
case to ensure that the solutions are real we require that the matrix $A$ is
real and consequently the vector $a\rangle$ is proportional to a real vector
$a\rangle=\hat{\gamma} \hat a\rangle$, where $\hat a \rangle\in
\mathbb{R}^{n+2}$. For real $\hat a\rangle$ it follows 
from the condition $\langle a a\rangle=0$ that $\hat a\rangle=0$ and so $A=0$. Thus
the case $\mu=\pm \nu, \ \nu\in \bbbr$ leads to a trivial result. 

We shall see that the case $ \mu=\pm i\nu, \ \nu\in \bbbr$ and
$\bar A=-Q A Q$ yields a nontrivial solution.
\end{Rem}

To get the kink solution similar to the one for the 
classical sine-Gordon equation, we take $\mu$ to be pure imaginary,
that is, $\mu=i \nu$, where $\nu\in \mathbb{R}$ as we stated in Remark \ref{rem1}. 

\begin{Pro}\label{pro4}
A kink solution of the vector sine-Gordon system (\ref{vsG})
on a trivial background ($\beta_0=1,\ \alpha^k_0=0$) is given by
\begin{eqnarray}\label{alpha0k}
\beta&=&1-2 (p_1^2-q_1^2)(p_1\cosh\rho+q_1\sinh\rho)^{-2} ,\\ 
\label{alphaik}
{\vec \alpha}&=&2  \ba_1
(p_1\sinh\rho+q_1\cosh\rho)(p_1\cosh\rho+q_1\sinh\rho)^{-2} ,\qquad
i=1,\ldots ,n\, .
\end{eqnarray}
Here $\rho=\nu x+\nu^{-1} t$ with $\nu\in \mathbb{R}$ and  $\langle a_1=
(i p_1,q_1, \ba_1^T)$, where $p_1, q_1 \in \mathbb{R}$ and
$\ba_1\in\mathbb{R}^n$  is a constant vector satisfying
$\langle a_1 a_1\rangle=0$. 
\end{Pro}
\begin{proof}
When $\mu=i\nu$, the fundamental solution $\Psi_0$ given by (\ref{psi0}) becomes
\begin{eqnarray} \label{psi0kink}
\Psi_0(x,t,i\nu)=\left(\begin{array}{ccc} \cosh \rho&
i\sinh \rho&0\\ -i\sinh \rho& \cosh \rho&0\\0&0&I\end{array}\right),\quad 
\rho=\nu x+\nu^{-1} t .
\end{eqnarray}
It follows from (\ref{veca}) that
\begin{eqnarray*}
 a\rangle=\Psi_0(x,t, i \nu) a_1\rangle=(i(p_1\cosh\rho+q_1\sinh\rho), 
p_1\sinh\rho+q_1\cosh\rho, \ba_1^T)^T.
\end{eqnarray*}
Thus
\begin{eqnarray*}
 \langle 
aQa\rangle=2(p_1\cosh\rho+q_1\sinh\rho)^2 .
\end{eqnarray*}
It follows from Proposition \ref{pro2} that
\begin{eqnarray*}
 b\rangle=\frac{i
\nu}{(p_1\cosh\rho+q_1\sinh\rho)^2}(-i(p_1\cosh\rho+q_1\sinh\rho), 
p_1\sinh\rho+q_1\cosh\rho, \ba_1^T)^T
\end{eqnarray*}
We can now write down the matrix $A$. Further we use 
formula (\ref{recv}) with $V_0=J$ and get (\ref{alphaik}) and (\ref{alpha0k})
for $\vec \alpha$ and $\beta$ as in the statement.
\end{proof}
One kink solution (\ref{alphaik}) belongs to a one dimensional subspace (it is 
proportional to the vector $\ba_1$) and thus it can be 
easily reconstructed from the kink solution of the scalar sine-Gordon equation. 
Indeed, if $\theta$ is a kink solution of equation $\theta_{xt}=\sin\theta$, then
\[\beta=\cos\theta,\qquad \vec{\alpha}=\frac{\ba_1}{|\ba_1|}\sin\theta\]
is the corresponding solution of the vector sine-Gordon equation. 

Later we 
shall show that a general two kink solution belongs to two dimensional subspace 
$\mbox{Span}_\bbbr(\ba_1,\ba_2)$ and it cannot be obtained from 
solutions of the scalar sine-Gordon equation.
In particular, for two-kink solutions we compute the position shifts after
collision.

For one kink solution in Proposition \ref{pro4}, we rewrite their denominator in
the
following form:
\begin{eqnarray}\label{posi}
(p_1\cosh\rho+q_1\sinh\rho)^2
=\left(\frac{p_1^2-q_1^2}{2}\right)(\cosh \left(2\nu (x+\nu^{-2} t-
x_0)\right)+1),\
x_0=\frac{1}{2\nu} \ln \big{|}\frac{p_1-q_1}{p_1+q_1}\big{|}\ .
\end{eqnarray}
Thus the kink is moving on the $(x, t)$-plane along the straight line and its
position at time $t$ is 
\begin{eqnarray*}
 x=-\nu^{-2} t+ x_0,
\end{eqnarray*}
which is only dependent on the pole $\mu=i \nu$  and the first two components of
vector 
$a_1\rangle$. The width of the kink is  of the order $\nu^{-1}$.

\subsection{A breather solution}\label{sec_br}

A breather solution corresponds to the only simple poles at points of a generic
orbit of the reduction group.
Assuming that the dressing matrix has a simple pole at a point
$\mu=\gamma+i\omega\, ,
\ \mu\ne \pm \bar{\mu}$, where we use $\bar \mu$ to denote the complex conjugate of $\mu$, we have
\begin{equation}
 \Phi(\lambda)=I_{n+2}+\frac{A}{\lambda-\mu}-\frac{QAQ}{\lambda+\mu}+\frac{\bar A}{
\lambda-{\bar \mu} } -\frac{Q{\bar A}Q}{\lambda+{\bar \mu}}\ .\label{Phi}
\end{equation}
Assume that $A=b\rangle\langle a\ne 0$. We determine the vector $b\rangle$
using Proposition \ref{pro00} and construct the solution on a trivial background ($\beta_0=1,\ {\vec \alpha}_0=0$).

\begin{Pro}\label{prob} Let $A=b\rangle\langle a\ne 0$ be a bi-vector, where 
$\langle a$ and $\langle b$ are two complex vectors given by
$\langle a=(p, q, a^1, \cdots, a^n)$ and $\langle b=(b^1, b^2,
\cdots, b^{n+2})$ respectively. The dressing matrix (\ref{Phi})
satisfies $\Phi(\lambda)\Phi^T(\lambda)=I_{n+2}$ if and only if
\begin{eqnarray}\label{aa}
 && \langle aa\rangle=0,\\ \label{b1}&& 
b^1=\frac{2i\gamma\omega\mu {\bar p}}{\Delta},\\ \label{b2}
 && b^2=\frac{2i\gamma\omega\mu}{\Delta^2}\left(
{\bar q} ({\bar \mu}^{2} \langle {\bar a} a \rangle+2\omega^2|p|^2)-4\gamma\omega
{\bar p}\im(p {\bar q})\right),\\ \label{bk}
 && b^{k+2}=\frac{2i\gamma\omega\mu}{\Delta^2}\left(
{\bar a}^k ({\bar \mu}^{2} \langle {\bar a} a \rangle+2\omega^2|p|^2)-4\gamma\omega
{\bar p} \im(p {\bar a}^k)\right),\quad k=1,\ldots,n\, ,
\end{eqnarray}
where
\begin{eqnarray}\label{defD}
 \Delta= |\mu|^2 \langle {\bar a} a \rangle-2\omega^2 |p|^2\, .
\end{eqnarray}
\end{Pro}

\begin{proof} Equation (\ref{aa}) has been proved in Proposition \ref{pro00}.
For matrix $A=b\rangle\langle a$ the second condition of (\ref{AkbP}) is
\begin{equation*}\label{Aeq}
 P a\rangle=(I_{n+2}-\frac{QAQ}{2\mu}+\frac{{\bar A}}{\mu-{\bar \mu}}
-\frac{Q{\bar A}Q}{\mu+{\bar \mu}})a\rangle=0\, ,
\end{equation*}
that is, 
\begin{equation}\label{ab}
a\rangle-\frac{Qb\rangle\langle
aQa\rangle}{2\mu}+\frac{{\bar b}\rangle\langle
{\bar a} a\rangle}{\mu-{\bar \mu}}
-\frac{Q{\bar b}\rangle\langle
{\bar a}Qa\rangle}{\mu+{\bar \mu}}=0.
\end{equation}
We solve this linear system of equations (\ref{ab}) for unknowns 
$b^1,\ldots 
,b^{n+2}$
by taking into account (\ref{aa}) and obtain the formulas in the statement.
\end{proof}


\begin{Pro}\label{pro8} A breather solution of the vector sine-Gordon system
(\ref{vsG})
on a trivial background ($\beta_0=1,\ {\vec \alpha}_0=0$) is given by
\begin{eqnarray}\label{alpha0b}
\beta&=&1-\frac{16\gamma^2\omega^2}{\Delta^2}\left( \langle {\bar a} a 
\rangle |q|^2-2(\im(p {\bar q}))^2\right)\, ,\\ \label{alphakb}
\alpha^k&=&\frac{8\gamma\omega}{\Delta^2} \left(\im(({\bar \mu}^{2}\langle {\bar
a} a 
\rangle +2\omega^2 |p|^2)q{\bar
a}^{k})-4\gamma\omega\im(q {\bar p})\im(p{\bar a}^k)\right)\, 
,\qquad
k=1,\ldots ,n\, ,
\end{eqnarray}
where $ \Delta$ is given by (\ref{defD}) and the 
components of the vectors $\langle a= (p,q,a^1,\ldots,
a^{n})$ are the following functions of $x,t$ given by
\begin{equation}\begin{array}{lll}
                 p&=&\cos (\phi +i \eta) p_1+\sin( \phi +i \eta)q_1,\\
q&=&\cos (\phi +i \eta) q_1-\sin( \phi +i \eta)p_1,\\
a^k&=&a_1^k, \qquad k=1,\ldots ,n,
                \end{array}\label{pqxt}
\end{equation}
where $\phi=\gamma x-\frac{\gamma}{(\gamma^2+\omega^2)}t, \ \eta=\omega
x+\frac{\omega}{(\gamma^2+\omega^2)}t$, vector  $\langle a_1=
(p_1,q_1,a_1^1,\ldots,a_1^{n})$ is an arbitrary complex
constant vector satisfying the condition $\langle a_1 a_1\rangle=0$.
\end{Pro}
\begin{proof}
When $\mu=\gamma+i\omega$, the fundamental solution $\Psi_0$ given by (\ref{psi0}) becomes
\begin{eqnarray}\label{psi0b} 
\Psi_0(\mu, x,t)=\left(\begin{array}{ccc} \cos (\phi +i \eta) &
i\sin (\phi +i \eta)&0\\ -i\sin (\phi +i \eta)& \cos (\phi +i \eta)&0\\0&0&I\end{array}\right),
\end{eqnarray}
where $\phi=\gamma \left(x-(\gamma^2+\omega^2)^{-1}t\right), \ \eta=\omega
\left(x+(\gamma^2+\omega^2)^{-1}t\right)$.
Using (\ref{veca}) we get (\ref{pqxt}) for $a\rangle$.

It follows from (\ref{recv}) that
\[
 V=\Phi(0) J \Phi^{T}(0),
\]
where $\Phi(\lambda)$ is defined in (\ref{Phi}). Making substitution of 
$A=b\rangle\langle a$ in (\ref{Phi}) and taking into account  (\ref{JUV}) we 
obtain
\begin{eqnarray}
&&\beta=V_{1,2}=\frac{1}{|\mu|^4}
(|\mu|^2-4\re({\bar \mu} p b^1))(|\mu|^2-4\re({\bar \mu} q b^2));\\
&&\alpha_k=V_{1,k+2}=-\frac{4}{|\mu|^4}
(|\mu|^2-4\re({\bar \mu} p b^1))\re({\bar \mu} q b^{k+2}),\quad k=1,\cdots, n.
\end{eqnarray}
From (\ref{b1}) in Proposition \ref{prob} it follows that  
\[
 \re({\bar \mu} p b^1)=\re\left({\bar \mu} p\frac{2i\gamma\omega\mu {\bar
p}}{\Delta}\right)=0.
\]
Thus $(|\mu|^2-4\re({\bar \mu} pb^1))=|\mu|^2$, and therefore
\begin{eqnarray}
\beta&=&1-\frac{4}{|\mu|^2}\re({\bar \mu} qb^2)\, ,\\ 
\alpha_k&=&-\frac{4}{|\mu|^2} \re({\bar \mu} qb^{k+2})\, ,\qquad
k=1,\ldots ,n\, .
\end{eqnarray}
The substitution of (\ref{b2}) and  (\ref{bk}) in the above equations leads to
(\ref{alpha0b}) and (\ref{alphakb}). 
\end{proof}

The breather solution (\ref{alpha0b}), (\ref{alphakb}) represents a
periodically 
oscillating localized wave of a characteristic width $\sim \omega^{-1}$  moving 
with the constant velocity $(\gamma^2+\omega^2)^{-1}$. Thus the width of the
wave, its speed 
and frequency of oscillations depend on the position of complex pole $\mu$ only.
The location of the wave  and the phase of its
oscillations 
are determined by $p_1,q_1$, i.e. the first two entries of the vector 
$a_1\rangle$. These can be found using the denominator $ 
\Delta$ defined by (\ref{defD}) as follows:
\begin{eqnarray*}
&&\Delta=|\mu|^2(\langle 
{\bar a}_1
a_1\rangle-|p_1|^2-|q_1|^2)+\gamma^2(|p_1|^2+|q_1|^2)\cosh(2\eta)-2\gamma^2 
\im({\bar p}_1 q_1)\sinh(2\eta) +\\
&&+\omega^2 
(|q_1|^2-|p_1|^2)\cos(2\phi)-2\omega^2\re(p_1\bar{q}_1)\sin(2\phi)\\
&&=(\gamma^2+\omega^2)\sum_{k=1}^n |a^k_1|^2
+|p_1^2+q_1^2|\left(\gamma^2\cosh 
(2\eta-2\eta_0)-\omega^2 \cos(2 \phi-\phi_0)\right), 
\end{eqnarray*}
where 
\begin{eqnarray}
\eta_0=\frac{1}{2}\ln \big{|}\frac{p_1-q_1 i}{p_1+q_1 i}\big{|}
\label{breapo}
\end{eqnarray}
and $\phi_0$ satisfies
$\cos \phi_0=\frac{|p_1|^2-|q_1|^2}{|p_1^2+q_1^2|}$ and $\sin \phi_0=\frac{2\re(p_1\bar{q}_1)}{|p_1^2+q_1^2|}$.
Thus we have
\begin{eqnarray}\label{breaph}
\phi_0=\arg \left(\frac{p_1+q_1 i}{p_1-q_1 i} \right)\ . 
\end{eqnarray}
Thus the breather is moving on the $(x, t)$-plane along the straight line and
its position at time $t$ is given by
\begin{eqnarray*}
 x=-(\gamma^2+\omega^2)^{-1} t+ x_0,\quad x_0=\omega^{-1} \eta_0 ,
\end{eqnarray*}
which is only dependent on the position of the pole $\mu$  and the first two
components of vector 
$a_1\rangle$.  We define the phase of this breather as $\phi_0$.
In Figure \ref{fig1}, we plot one breather solution when $n=2$ for
$\alpha_1$ and $\alpha_2$ with $\mu=2+\frac{1}{4} i$ and
$\langle a_1=(i, 2, 3i, \sqrt{6})$ at time $t=3$. The position of this breather 
is
at $x\simeq-2.936$ and its phase $\phi_0=0$.

The breather solution (\ref{alpha0b}), (\ref{alphakb}) also depends on the 
constant complex $n$ dimensional vector  $\ba_1=(a_1^1,\ldots, a_1^n)^T\in 
\C^n$, which is the last $n$ components of the vector $ 
a_1\rangle=(p_1,q_1,a_1^1,\ldots, a_1^n)^T$. In general, the vector
$\vec{\alpha}$ 
of 
the solution (\ref{alphakb}) belongs to the two dimensional subspace 
$W(\ba_1)=\mbox{Span}_\bbbr(\re \ba_1,\im\ba_1)\subseteq\bbbr^n$. If the 
subspace $W(\ba_1)$ is one dimensional ($\re \ba_1\sim \im\ba_1$) 
then the obtained one breather solution represents  the well known breather 
solution of the scalar sine-Gordon equation (with 
$\theta_{xt}=\sin\theta:\ \beta=\cos\theta, \
\vec{\alpha}=\tfrac{\ba_1}{|\ba_1|}\sin\theta$), which 
is 
a subsystem of the vector sine-Gordon equation. Solution  (\ref{alphakb}) 
corresponding to a two-dimensional subspace $W(\ba)$ cannot be obtained 
from a solution of the scalar sine-Gordon equation. However, it can be obtained
from the breather solution of the two-dimensional vector sine-Gordon equation.

Similar to the kink case (see Remark \ref{rem0}), the breather solution is
parametrised by a point on the complex Grassmanian ${\rm
Gr}_{1,2+n}=\mathbb{CP}^{n+1}$ rather than the complex vector $a\rangle$ itself.

\section{Multi-soliton solutions and soliton interactions}

The vector sine-Gordon equation  has two types of one  soliton solutions,
namely 
the kink and breather solutions discussed in sections \ref{Kink} and
\ref{sec_br}. The set of
poles of the dressing matrix $\Phi(\lambda)$ for multi-soliton solutions is a 
union of a finite number of the kink and breather type orbits.   
We can construct the dressing matrix recursively by presenting it as a composition 
of elementary dressing factors of the kink and breather type.  Another approach 
is based on a representation of the dressing matrix in 
the form of partial fractions with respect to the spectral parameter 
$\lambda$. The latter approach is more conventional, but leads to a big system 
of linear algebraic equations, the size of which is determined by the number 
of solitons in the multi-soliton solution.

In order to construct a general multi-soliton dressing of the vector
sine-Gordon 
equation (\ref{vsG}) with $m_1$ kinks and $m_2$ breathers we need to choose the
following 
data:
\begin{itemize}
 \item a set of $m_1$ distinct positive real numbers 
$\{\nu_s\}_{s=1}^{m_1}$,
 \item a set of $m_1$ vectors $a_s\rangle,\ s=1,\ldots, m_1$ of the form 
$a_s\rangle=(ip_s,q_s,a^1_s,\ldots, a^n_s)^T$, such that  
$$\langle 
a_s a_s\rangle=-p_s^2+q_s^2+(a^1_s)^2+\cdots +(a^n_s)^2=0, \qquad  
p_s,q_s,a^1_s,\ldots,a^n_s\in\bbbr$$
and the real ``sub-vector'' $(a^1_s,\ldots, 
a^n_s)$ is non-zero.
\item a set of $m_2$ distinct complex  numbers 
$\{\mu_s=\gamma_s+i\omega_s\,|\, 
,\gamma_s>0,\ \omega_s>0 \}_{ s=m_1+1 } ^ {m_1+m_2},\ $
 \item a set of $m_2$ complex vectors
$a_s\rangle=(p_s,q_s,a^1_s,\ldots, 
a^n_s)^T$, $s=m_1+1, \ldots, m_1+m_2$ and $ p_s, q_s, a^1_s, \ldots, 
a^n_s\in\C$ such that  $\langle 
a_s a_s\rangle=0$ and 
``sub-vectors'' $(p_s,q_s)$ and $(a^1_s,\ldots, a^n_s)$ are both non-zero.
\end{itemize}
Let us first construct the $m_1$-kink solution.
We  start from a trivial solution of the vector sine-Gordon equation 
$\beta_0=1,\vec{\alpha}_0=0$, so that  
\[ 
 U_0=0,\quad V_0=J, \quad \Psi_0(\lambda,x,t)=\exp(\lambda x J-\lambda^{-1} t 
J).
\]
We denote the kink dressing matrix constructed in the Section \ref{Kink} as
\begin{equation*}
\widehat \Phi(\lambda, 
\mu,a\rangle)=I_{n+2}+\frac{1}{(\lambda-\mu)}\, \frac{Qa\rangle\langle 
a}{\langle aQa\rangle}-\frac{1}{(\lambda+\mu)}\, \frac{a\rangle\langle 
aQ}{\langle aQa\rangle}.
\end{equation*}
Now the multi-kink solution of the vector sine-Gordon equation and the 
corresponding the fundamental solution for the associated linear problems 
can be found recursively
\begin{equation*} \begin{array}{l}        
V_s=\widehat \Phi(0,i\nu_s,\Psi_{s-1}(i\nu_s,x,t)a_s\rangle)V_{s-1}
\widehat \Phi^T(0, i\nu_s, \Psi_ { s-1 }(i\nu_s,x,t) a_s\rangle);\\
\Psi_s(\lambda,x,t)=\widehat \Phi(\lambda,i\nu_s,\Psi_{s-1}(i\nu_s,x,
t)a_s\rangle)\Psi_{ s-1 }(\lambda,x,t);
                 \end{array}\qquad s=1, 2, ..., m_1 .
\end{equation*}
Having constructed the $m_1$-kink solution $V_{m_1}$ and the
corresponding fundamental solution of the linear problem $\Psi_{m_1}(\lambda, x,
t)$ we can add $m_2$ breathers using the breather dressing matrix constructed
in Section \ref{sec_br} (see Proposition \ref{prob}) denoted as
\begin{equation*}
\widetilde \Phi(\lambda,\mu, a\rangle) 
=I_{n+2}+\frac{A}{\lambda-\mu}-\frac{QAQ}{\lambda+\mu}+\frac{ \bar A}{
\lambda-{\bar \mu} } -\frac{Q{\bar A}Q}{\lambda+{\bar \mu}}
\end{equation*}
in a similar recursive way
\begin{equation*} \begin{array}{l}        
V_s=\widetilde \Phi(0,\mu_s,\Psi_{s-1}(\mu_s,x,t)a_s\rangle)V_{s-1}
\widetilde \Phi^T(0, \mu_s, \Psi_ { s-1 }(\mu_s,x,t) a_s\rangle);\\
\Psi_s(\lambda,x,t)=\widetilde \Phi(\lambda,\mu_s,\Psi_{s-1}(\mu_s,x,
t)a_s\rangle)\Psi_{ s-1 }(\lambda,x,t);
                 \end{array}\qquad s=m_1+1, ..., m_1+m_2 .
\end{equation*}
Moreover, we could change the order of dressings by creating $m_2$-breather
solutions and then adding $m_1$ kinks, or even make dressings in an arbitrary
order. Such recursive approach is useful for the study of the effects
of soliton interactions as shown in the following sections.

\subsection{Interaction of two kinks}\label{interk}
In this section, we study the interaction of two kinks.
A two-kink solution corresponding to the sets 
$$\{\nu_1,a_1\rangle=(i p_1,q_1,a_1^1,\ldots ,a_1^n)^T\}, \quad
\{\nu_2,a_2\rangle=(i p_2,q_2,a_2^1,\ldots ,a_2^n)^T\},\quad 0<\nu_1<\nu_2$$ 
represents two kinks moving to the left with speeds $\nu_1^{-2}$ and  
$\nu_2^{-2}$ respectively. Their trajectories on the $(x,t)$--plane intersect 
near the point $(X,T)$, which can be found from the system of equations (c.f. 
(\ref{posi}))
\[
 \nu_1 X+\nu_1^{-1}T=\frac{1}{2}\ln 
\big{|}\frac{p_1-q_1}{p_1+q_1}\big{|},\qquad \nu_2 
X+\nu_2^{-1}T=\frac{1}{2}\ln \big{|}\frac{p_2-q_2}{p_2+q_2}\big{|}\, .
\]
Far from this point as $t\to\pm\infty$ the solution tends to the sum of two 
simple kinks moving along the trajectories
\[
 \nu_1 (x-x_1^\pm)+\nu_1^{-1}t=0,\qquad  \nu_2 (x-x_2^\pm)+\nu_2^{-1}t=0
\]
respectively. Thus, the effect of the interaction is a shift of a straight line 
trajectory of each kink. In Figure \ref{fig2}, we plot a two-kink solutions
when $n=2$ for $\alpha_1$ and $\alpha_2$ at time 
$t=6$ with
\[
\nu_1=0.9,\qquad \langle a_1=(i, 0, 0, 1);\qquad
\nu_2=1.1,\qquad \langle a_2=(i, 0, 1, 0).
\]
These two kinks intersect at the point $x=0$. At $t=6$, the position of the
kink corresponding to $\nu_1$ is at $x=\simeq -7.41$ and the other
kink is at $x=-4.96$.

In order to determine the shift $x_2=x_2^+-x_2^-$ of the trajectory of the 
second kink we consider the limits $t\to\pm\infty$ assuming $\rho_2=\nu_2 
x+\nu_2^{-1}t$ to be finite. In this case 
\[
 \rho_1=\nu_1 
x+\nu_1^{-1}t=\frac{\nu_1}{\nu_2}\rho_2+\frac{t}{\nu_1}\left(1-\frac{\nu_1^2}{
\nu _2^2}\right)\to\pm\infty,\quad \mbox{as }\ t\to\pm\infty.
\]
Therefore 
\begin{equation}
  \Psi_1(\lambda,x,t)a_2\rangle=\widehat \Phi(\lambda,i\nu_1,\Psi_0(i\nu_1,x,
t)a_1\rangle)\Psi_{0}(\lambda,x,t)a_2\rangle\to \widehat\Phi_1^\pm (\lambda)
\Psi_{0}(\lambda,x,t)a_2,\quad \mbox{as }\ t\to\pm\infty,\label{psi1pm}
\end{equation}
where 
\[
 \widehat\Phi_1^\pm (\lambda)=\lim_{\rho_1\to\pm\infty}\widehat 
\Phi(\lambda,i\nu_1,\Psi_0(i\nu_1,x,
t)a_1\rangle).
\]
The latter limit can be easily computed since $\Psi_0(i\nu_1,x,
t)$ is a function of $\rho_1$ only (\ref{psi0kink}). In particular we have 
\[
 \lim_{\rho_1\to\pm\infty}\frac{Q\Psi_0(i\nu_1,x,
t)a_1\rangle\langle a_1 \Psi_0(i\nu_1,x,
t)^T}{\langle a_1 \Psi_0(i\nu_1,x,
t)^TQ\Psi_0(i\nu_1,x,
t)a_1\rangle}=\frac{1}{2}\left(\begin{array}{ccc} 
1&\mp i&0\\\pm i&1&0\\0&0&0 \end{array} \right)
\]
and therefore 
\begin{eqnarray}\label{phibki}
\widehat\Phi_1^\pm (\lambda)=\frac{1}{\lambda^2+\nu_1^2}\left(\begin{array}{ccc}
\lambda^2-\nu_1^2 & \pm 2\lambda \nu_1&0\\\mp 2\lambda 
\nu_1&\lambda^2-\nu_1^2&0\\0&0&(\lambda^2+\nu_1^2)I
\end{array}\right).
\end{eqnarray}
Noting that $[\widehat\Phi_1^\pm (\lambda),\Psi_0(\lambda,x,
t)]=0$ we can represent the limits (\ref{psi1pm}) in the form
\[
 \widehat\Phi_1^\pm (i\nu_2)
\Psi_{0}(i\nu_2,x,t)a_2\rangle=
\Psi_{0}(i\nu_2,x,t)a_2^\pm\rangle
\]
where 
\[ a_2^\pm\rangle=\left(\begin{array}{c}
i(\gamma_1p_2\mp\gamma_2 
q_2)\\ 
\gamma_1 q_2\mp\gamma_2 
p_2\\
 \ba_2
\end{array}\right),\qquad 
\gamma_1=\frac{\nu_2^2+\nu_1^2}{\nu_2^2-\nu_1^2} , \qquad 
\gamma_2=\frac{2 \nu_1\nu_2}{\nu_2^2-\nu_1^2}.
\]
Using formula (\ref{posi}) we obtain
\begin{eqnarray*}
&& x_2^{+}=\lim_{\rho\rightarrow +\infty}x_0=\frac{1}{2 \nu_2}
\ln \big{|}\frac{(\gamma_1+\gamma_2)(p_2-
q_2)}{(\gamma_1-\gamma_2)(p_2+ q_2)}\big{|};\\
&& 
x_2^{-}=\lim_{\rho\rightarrow -\infty}x_0=\frac{1}{2 \nu_2}
\ln \big{|}\frac{(\gamma_1-\gamma_2)(p_2-q_2)}{(\gamma_1+\gamma_2)(p_2+q_2)}\big{|}.
\end{eqnarray*} 
Therefore, the position shift for the kink with speed $1/\nu_2^2$ is
\begin{eqnarray*}
\Delta x_2=x_2^{+}-x_2^{-}=\frac{1}{2\nu_2}\ln
\left(\frac{\gamma_1+\gamma_2}{\gamma_1-\gamma_2}\right)^2=\frac{1}{\nu_2}\ln
\left(\frac{\nu_1+\nu_2}{\nu_1-\nu_2}\right)^2
=\frac{2}{\nu_2}\ln
\frac{\nu_2+\nu_1}{\nu_2-\nu_1} .
\end{eqnarray*}
In a similar way one can obtain that the position shift for the kink with speed
$1/\nu_1^2$ is
\begin{eqnarray*}
\Delta x_1=x_1^{+}-x_1^{-}
=-\frac{1}{\nu_1}
\ln\left(\frac{\nu_1+\nu_2}{\nu_2-\nu_1}\right)^2=\frac{2}{\nu_1}
\ln\frac{\nu_2-\nu_1}{\nu_2+\nu_1}  .
\end{eqnarray*}
Notice that these position shifts $\Delta x_1$ and $\Delta x_2$ for kinks after
collision only depend on the pole position $\nu_1$ and $\nu_2$. They are
exactly the same as that of the  scalar 
sine-Gordon equation (\ref{ssg}) given in \cite{fad87}.

\subsection{Interaction of two breathers}\label{interb}

A two-breather solution corresponding to the sets 
$$\{\mu_1,a_1\rangle=(p_1,q_1,a_1^1,\ldots ,a_1^n)^T\}, \quad
\{\mu_2,a_2\rangle=(p_2,q_2,a_2^1,\ldots ,a_2^n)^T\},\quad 0<|\mu_1|<|\mu_2|$$ 
represents two breathers moving to the left with speeds $|\mu_1|^{-2}$ and  
$|\mu_2|^{-2}$ respectively. Their trajectories on the $(x,t)$--plane intersect 
near the point $(X,T)$, which can be found from the system of equations (c.f. 
(\ref{breapo}))
\begin{eqnarray*}
 \omega_l X+\omega_l (\gamma_l^2+\omega_l^2)^{-1}T=\frac{1}{2}\ln
\big{|}\frac{p_l-q_l i}{p_l+q_l i}\big{|}, \quad l=1, 2.
\end{eqnarray*}
Far from this point as $t\to\pm\infty$ the solution tends to the sum of two 
simple breathers moving along the trajectories
\[
  x-x_l^\pm+(\gamma_l^2+\omega_l^2)^{-1}t=0,\qquad  l=1,
2.
\]
with phase $\phi_l^\pm$.
Thus, the effect of the interaction is a shift of a straight line 
trajectory and the change of the phase for each breather. 

In order to determine the shift $\eta_2=\eta_2^+-\eta_2^-$ of the trajectory
and the change of the phase for
the second breather we consider the limits $t\to\pm\infty$ assuming
$\rho_2= x+(\gamma_2^2+\omega_2^2)^{-1}t$ to be finite. In
this case 
\[
 \rho_1= x+(\gamma_1^2+\omega_1^2)^{-1} t
=\rho_2+
\left(\frac{1}{\gamma_1^2+\omega_1^2}-\frac{1}{\gamma_2^2+\omega_2^2}\right)
t\to\pm\infty,\quad \mbox{as }\ t\to\pm\infty.
\]
Therefore 
\begin{equation}
  \Psi_1(\lambda,x,t)a_2\rangle=\widetilde \Phi(\lambda,\mu_1,\Psi_0(\mu_1,x,
t)a_1\rangle)\Psi_{0}(\lambda,x,t)a_2\rangle\to \widetilde\Phi_1^\pm (\lambda)
\Psi_{0}(\lambda,x,t)a_2,\quad \mbox{as }\ t\to\pm\infty,\label{psi2pm}
\end{equation}
where 
\[
 \widetilde\Phi_1^\pm (\lambda)=\lim_{\rho_1\to\pm\infty}\widetilde 
\Phi(\lambda,\mu_1,\Psi_0(\mu_1,x,t)a_1\rangle).
\]
The latter limit can be easily computed using $\Psi_0(\mu_1,x, t)$ 
(\ref{psi0b}). In particular we have 
\[
 \lim_{\rho_1\to\pm\infty}A=\frac{\omega_1
\mu_1}{\gamma_1}\left(\begin{array}{ccc} 
i&\pm1&0\\\mp1&i&0\\0&0&0 \end{array} \right)
\]
and therefore 
\begin{eqnarray}\label{phibr}
\widetilde\Phi_1^\pm (\lambda)
=\left(\begin{array}{ccc}\frac{(\lambda^2-|\mu_1|^2)^2-4 \lambda^2
\omega_1^2}{(\lambda^2-\mu_1^2)(\lambda^2-{\bar \mu}_1^2)}
&\pm\frac{4 \lambda \omega_1
(\lambda^2-|\mu_1|^2)}{(\lambda^2-\mu_1^2)(\lambda^2-{\bar \mu}_1^2)}&0\\
\mp \frac{4 \lambda \omega_1
(\lambda^2-|\mu_1|^2)}{(\lambda^2-\mu_1^2)(\lambda^2-{\bar \mu}_1^2)}
& \frac{(\lambda^2-|\mu_1|^2)^2-4 \lambda^2
\omega_1^2}{(\lambda^2-\mu_1^2)(\lambda^2-{\bar \mu}_1^2)}&0\\
0&0&I \end{array} \right).
\end{eqnarray}
Noting that $[\widetilde\Phi_1^\pm (\lambda),\Psi_0(\lambda,x,
t)]=0$ we can represent the limits (\ref{psi2pm}) in the form
\[
 \widetilde\Phi_1^\pm (\mu_2)
\Psi_{0}(\mu_2,x,t)a_2\rangle=
\Psi_{0}(\mu_2,x,t)a_2^\pm\rangle .
\]
 We denote $(1,1)$ and $(1,2)$ entries of matrix $\widetilde\Phi_1^+ (\mu_2)$ by
$\kappa_1$ and $\kappa_2$, respectively. Notice that
\begin{eqnarray*}
&&\kappa_1-\kappa_2 i=\frac{(\mu_2^2-|\mu_1|^2-2 \mu_2
\omega_1 i)^2}{(\mu_2^2-\mu_1^2)(\mu_2^2-{\bar \mu}_1^2)}
=\frac{(\mu_2-\mu_1)^2 (\mu_2+\bar \mu_1)^2}{(\mu_2^2-\mu_1^2)(\mu_2^2-{\bar
\mu}_1^2)};\\
&&\kappa_1+\kappa_2 i=\frac{(\mu_2^2-|\mu_1|^2+2 \mu_2
\omega_1 i)^2}{(\mu_2^2-\mu_1^2)(\mu_2^2-{\bar \mu}_1^2)}
=\frac{(\mu_2+\mu_1)^2 (\mu_2-\bar \mu_1)^2}{(\mu_2^2-\mu_1^2)(\mu_2^2-{\bar
\mu}_1^2)}
\end{eqnarray*}
Using formula for $\eta_0$ in (\ref{breapo}) we obtain
\begin{eqnarray*}
&& x_2^{+}=\frac{1}{\omega_2}\lim_{\rho_1\rightarrow+\infty}\eta_0
=\frac{1}{2\omega_2}
\ln \big{|}\frac{(\kappa_1+\kappa_2 i) (p_2-q_2 i)}
{\kappa_1-\kappa_2 i) (p_2+q_2 i)}\big{|}
=\frac{1}{2\omega_2}
\ln \big{|}\frac{(\mu_1+\mu_2)^2 (\bar \mu_1- \mu_2)^2 (
p_2-q_2 i )}{(\mu_1-\mu_2)^2 (\bar \mu_1+ \mu_2)^2 ( p_2+q_2 i)}\big{|};\\
&& 
x_2^{-}=\frac{1}{\omega_2}\lim_{\rho_1\rightarrow -\infty}\eta_0
=\frac{1}{2\omega_2}
\ln \big{|}\frac{(\kappa_1-\kappa_2 i) (p_2-q_2 i)}
{\kappa_1+\kappa_2 i) (p_2+q_2 i)}\big{|}
=\frac{1}{2\omega_2}
\ln \big{|}\frac{(\mu_1-\mu_2)^2 (\bar \mu_1+ \mu_2)^2 (
p_2-q_2 i )}{(\mu_1+\mu_2)^2 (\bar \mu_1- \mu_2)^2 ( p_2+q_2 i)}\big{|}.
\end{eqnarray*}
Therefore, the position shift for the breather with speed $1/|\mu_2|^2$ is
\begin{eqnarray*}
\Delta x_2=x_2^{+}-x_2^{-}=\frac{2}{\omega_2}\ln \big{|}\frac{(\mu_1+\mu_2)
(\bar \mu_1- \mu_2)}{(\mu_1-\mu_2) (\bar \mu_1+ \mu_2)} \big{|}.
\end{eqnarray*}
Using the formula for $\phi_0$ in (\ref{breaph}), we compute the phase shift
for this breather. Indeed, we have
\begin{eqnarray*}
&& \phi_2^{+}=\lim_{\rho_1\rightarrow +\infty}\phi_0= \arg \left( 
\frac{(\mu_1-\mu_2)^2 (\bar \mu_1+\mu_2)^2 (
p_2+q_2 i )}{(\mu_1+\mu_2)^2 (\bar \mu_1- \mu_2)^2 ( p_2-q_2 i)}
 \right);
 \\&&
\phi_2^{-}=\lim_{\rho_1\rightarrow +\infty}\phi_0= \arg \left( 
\frac{(\mu_1+\mu_2)^2 (\bar \mu_1-\mu_2)^2 (
p_2+q_2 i )}{(\mu_1-\mu_2)^2 (\bar \mu_1+ \mu_2)^2 ( p_2-q_2 i)} \right) \ .
\end{eqnarray*} 
Hence the phase shift is
\begin{eqnarray*}
&&\Delta\phi_2=\phi_2^{+}-\phi_2^{-}= 2\arg \left( 
\frac{(\mu_1-\mu_2)^2 (\bar \mu_1+\mu_2)^2} 
{(\mu_1+\mu_2)^2 (\bar \mu_1-\mu_2)^2}\right)
=4 \arg \left( \frac{(\mu_1-\mu_2)(\bar \mu_1+ \mu_2)}{ ( \mu_1+
\mu_2)(\bar \mu_1- \mu_2)} \right)
\ .
\end{eqnarray*}
In a similar way we obtain that the position and phase shifts for the breather
with speed $1/|\mu_1|^2$ are
\begin{eqnarray*}
&&\Delta x_1=x_1^{+}-x_1^{-}=\frac{2}{\omega_1}\ln
\big{|}\frac{(\mu_1-\mu_2) (\bar\mu_1+ \mu_2)}{(\mu_1+\mu_2) (\bar \mu_1-
\mu_2)} \big{|};\\
&&\Delta\phi_1=\phi_1^{+}-\phi_1^{-}=4 \arg \left(
\frac{(\mu_1+\mu_2)(\bar \mu_1- \mu_2)}{ ( \mu_1-
\mu_2)( \bar \mu_1+\mu_2)} \right)\ .
\end{eqnarray*}
Similar to the case of interaction of two kinks,  the position and phase shifts
for breathers after collision only depend on the pole position $\mu_1$ and
$\mu_2$, which are exactly the same as that of the  scalar 
sine-Gordon equation (\ref{ssg}) given in \cite{fad87}.

\subsection{Interaction of one kink and one breather}

A kink-breather or breather-kink solution corresponding to the sets 
$$\{i \nu, a\rangle=(i \hat p, \hat q, \hat a^1,\ldots , \hat a^n)^T\}, \quad
\{\mu,\tilde a \rangle=(\tilde p,\tilde q,\tilde a^1,\ldots ,\tilde
a^n)^T\},\quad 0<\nu<|\mu|,$$ 
where $\nu, \hat p, \hat q, \hat a^1,\ldots , \hat a^n \in \mathbb{R}$ and
$\mu=\gamma +i \omega, \tilde p, \tilde q, \tilde a^1, \ldots , \tilde a^n\in
\mathbb{C}$, 
represents one kink and one breather moving to the left with speeds $\nu^{-2}$
and 
$|\mu|^{-2}$ respectively. Their trajectories on the $(x,t)$--plane intersect 
near the point $(X,T)$, which can be found from the system of equations (cf. 
(\ref{posi}))
\[
 \nu X+\nu^{-1}T=\frac{1}{2}\ln 
\big{|}\frac{\hat p-\hat q}{\hat p+\hat q}\big{|},\qquad 
 \omega X+\omega (\gamma^2+\omega^2)^{-1}T=\frac{1}{2}\ln
\big{|}\frac{\tilde p-\tilde q i}{\tilde p+\tilde q i}\big{|} \ .
\]
Far from this point as $t\to\pm\infty$ the solution tends to the sum of one
kink and one breather moving along the trajectories
\[
 \nu (x-\hat x^\pm)+\nu^{-1}t=0,\qquad 
(x-\tilde x^\pm)+(\gamma^2+\omega^2)^{-1}t=0
\]
respectively, and the phase for the breather is $\phi^\pm$. Thus, the effect of
the interaction is a shift of a straight line 
trajectory of each of them.

In order to determine the shift $\Delta\tilde x=\tilde x^+-\tilde x^-$ of the
trajectory and phase shift $\Delta\phi=\phi^+-\phi^-$ of
the breather we consider the limits $t\to\pm\infty$ assuming
$\tilde \rho= x+(\gamma^2+\omega^2)^{-1}t$ to be finite. In
this case 
\[
 \hat \rho= \nu x+\nu^{-1} t
=\nu \left( \tilde \rho+\left(\frac{1}{\nu^2}-\frac{1}{\gamma^2+\omega^2}\right)
t \right) \to\pm\infty,\quad \mbox{as }\ t\to\pm\infty.
\]
Therefore 
\begin{equation}
  \Psi_1(\lambda,x,t)\tilde a\rangle=\widehat
\Phi(\lambda,i \nu,\Psi_0(i\nu,x,t) \hat a\rangle)
\Psi_{0}(\lambda,x,t) \tilde a\rangle\to \widehat\Phi_1^\pm (\lambda)
\Psi_{0}(\lambda,x,t) \tilde a,\quad \mbox{as }\ t\to\pm\infty,\label{psi3pm}
\end{equation}
where 
\[
 \widehat\Phi_1^\pm (\lambda)=\lim_{\hat \rho\to\pm\infty}\widehat 
\Phi(\lambda,i \nu,\Psi_0(i \nu,x,t)\hat
a\rangle)=\frac{1}{\lambda^2+\nu^2}\left(\begin{array}{ccc}
\lambda^2-\nu^2 & \pm 2\lambda \nu&0\\\mp 2\lambda 
\nu&\lambda^2-\nu^2&0\\0&0&(\lambda^2+\nu^2)I
\end{array}\right).
\]
according to (\ref{phibki}).
Noting that $[\widehat\Phi_1^\pm (\lambda),\Psi_0(\lambda,x,
t)]=0$ we represent the limits (\ref{psi3pm}) in the form
\[
 \widehat\Phi_1^\pm (\mu)
\Psi_{0}(\mu,x,t)\tilde a\rangle=
\Psi_{0}(\mu,x,t)\tilde a^\pm\rangle
\]
where
\[ \langle \tilde a^\pm=\left(\frac{\mu^2-\nu^2}{\mu^2+\nu^2} \tilde p\pm
\frac{2\mu \nu}{\mu^2+\nu^2} \tilde q,
\mp\frac{2\mu \nu}{\mu^2+\nu^2}\tilde p+\frac{\mu^2-\nu^2}{\mu^2+\nu^2} \tilde
q, \tilde a^1,\ldots ,\tilde a^n \right)
\]
Using formula (\ref{breapo}) we obtain
\begin{eqnarray*}
&& \tilde x^{+}=\frac{1}{\omega}\lim_{\hat \rho\rightarrow
+\infty}\eta_0=\frac{1}{2\omega}
\ln\frac{((\nu+\omega)^2+\gamma^2) |\tilde p-\tilde q i|}
{((\nu-\omega)^2+\gamma^2) |\tilde p+\tilde q i|};\\
&& 
\tilde x^{-}=\frac{1}{\omega}\lim_{\hat \rho\rightarrow
-\infty}\eta_0=\frac{1}{2\omega}\ln\frac{((\nu-\omega)^2+\gamma^2)
|\tilde p-\tilde q i|}
{((\nu+\omega)^2+\gamma^2)  |\tilde p+\tilde q i|}.
\end{eqnarray*} 
It follows that the position shift for the breather is
\begin{eqnarray*}
 \Delta\tilde x=\tilde x^{+}-\tilde x^{-}
=\frac{2}{\omega}\ln \big{|} \frac{\mu+\nu i}{\mu-\nu i} \big{|} 
=\frac{1}{\omega}\ln
\frac{(\nu+\omega)^2+\gamma^2}{(\nu-\omega)^2+\gamma^2} .
\end{eqnarray*}
Using formula (\ref{breaph}) we obtain
\begin{eqnarray*}
&&\phi^{+}=\lim_{\hat \rho\rightarrow +\infty}\phi_0= \arg\left(
\frac{(\mu-\nu i)^2(\tilde p+\tilde q i)}{(\mu+\nu i)^2(\tilde p-\tilde q i)}
\right);\quad 
\phi^{-}=\lim_{\hat \rho\rightarrow -\infty}\phi_0= \arg\left(
\frac{(\mu+\nu i)^2(\tilde p+\tilde q i)}{(\mu-\nu i)^2(\tilde p-\tilde q i)}
\right).
\end{eqnarray*}
Hence the phase shift for the breather is
\begin{eqnarray*}
\Delta\phi=\phi^{+}-\phi^{-}=4 \arg\left(
\frac{\mu-\nu i}{\mu+\nu i}
\right)
\end{eqnarray*}
which are only dependent on the positions of poles, the values of $\nu$ and
$\mu$.

In order to determine the shift $\Delta \hat x=\hat x^+-\hat x^-$ of the
trajectory of
the kink we consider the limits $t\to\pm\infty$ assuming
$\hat \rho= \nu x+\nu^{-1} t$ to be finite. In
this case 
\[
 \tilde \rho= x+(\gamma^2+\omega^2)^{-1}t
=\frac{\hat \rho}{\nu}-\left(\frac{1}{\nu^2}-\frac{1}{\gamma^2+\omega^2}
t \right) \to\mp\infty,\quad \mbox{as }\ t\to\pm\infty.
\]
Therefore 
\begin{equation}
  \Psi_1(\lambda,x,t)\hat a\rangle=\widetilde
\Phi(\lambda,\mu,\Psi_0(\mu,x,t) \tilde a\rangle)
\Psi_{0}(\lambda,x,t) \hat a\rangle\to \widetilde\Phi_1^\pm (\lambda)
\Psi_{0}(\lambda,x,t) \hat a,\quad \mbox{as }\ t\to\pm\infty,\label{psi4pm}
\end{equation}
where 
\[
 \widetilde\Phi_1^\pm (\lambda)=\lim_{\tilde \rho\to\mp\infty}\widetilde 
\Phi(\lambda, \mu,\Psi_0( \mu,x,t)\tilde
a\rangle)=\left(\begin{array}{ccc}\frac{(\lambda^2-|\mu|^2)^2-4 \lambda^2
\omega^2}{(\lambda^2-\mu^2)(\lambda^2-{\bar \mu}^2)}
&\mp\frac{4 \lambda \omega
(\lambda^2-|\mu|^2)}{(\lambda^2-\mu^2)(\lambda^2-{\bar \mu}^2)}&0\\
\pm \frac{4 \lambda \omega
(\lambda^2-|\mu|^2)}{(\lambda^2-\mu^2)(\lambda^2-{\bar \mu}^2)}
& \frac{(\lambda^2-|\mu|^2)^2-4 \lambda^2
\omega^2}{(\lambda^2-\mu^2)(\lambda^2-{\bar \mu}^2)}&0\\
0&0&I \end{array} \right).
\]
according to (\ref{phibr}).
Noting that $[\widetilde\Phi_1^\pm (\lambda),\Psi_0(\lambda,x,
t)]=0$ we represent the limits (\ref{psi4pm}) in the form
\[
 \widetilde\Phi_1^\pm (i \nu)
\Psi_{0}(i \nu,x,t)\hat a\rangle=
\Psi_{0}(i \nu,x,t)\hat a^\pm\rangle
\]
where
\[ \hat a^\pm\rangle=\left(\begin{array}{c}
i\frac{\left( (\gamma^2+\omega^2+\nu^2)^2+4
\nu^2 \omega^2\right)\hat p\pm 4\nu \omega(\gamma^2+\omega^2+\nu^2)\hat q
}{(\gamma^2+\omega^2+\nu^2)^2-4 \nu^2 \omega^2}\\
\frac{\pm 4 \nu \omega(\gamma^2+\omega^2+\nu^2)\hat p+\left(
(\gamma^2+\omega^2+\nu^2)^2+4 \nu^2 \omega^2\right) \hat q
}{(\gamma^2+\omega^2+\nu^2)^2-4 \nu^2 \omega^2} \\
\hat \ba
\end{array}
\right)
\]
Using formula (\ref{posi}) we obtain
\begin{eqnarray*}
&&\hat x^{+}=\lim_{\tilde\rho\rightarrow -\infty}x_0=\frac{1}{2 \nu}
\ln\frac{((\nu-\omega)^2+\gamma^2)^2 |\hat p-\hat
q|}{((\nu+\omega)^2+\gamma^2)^2 |\hat p+\hat q|};\\
&& 
\hat x^{-}=\lim_{\tilde\rho\rightarrow +\infty}x_0=\frac{1}{2 \nu}
\ln\frac{((\nu+\omega)^2+\gamma^2)^2 |\hat p-\hat
q|} {((\nu-\omega)^2+\gamma^2)^2 |\hat p+\hat q|}.
\end{eqnarray*} 
Therefore, the position shift for the kink with speed $1/\nu^2$ is
\begin{eqnarray*}
\Delta \hat x=\hat x^{+}-\hat x^{-}=\frac{1}{2 \nu}
\ln\frac{((\nu-\omega)^2+\gamma^2)^4 }{((\nu+\omega)^2+\gamma^2)^4 }
=\frac{2}{\nu}
\ln\frac{(\nu-\omega)^2+\gamma^2}{(\nu+\omega)^2+\gamma^2}.
\end{eqnarray*}
In Figure \ref{fig3}, we give the contour plot of $\alpha_1$ to in $(x,t)$-plane
to demonstrate the phase and position shifts of one kink and one breather
interaction, where the kink and breather are corresponding to 
\[
\nu=1,\ \langle a_{\rm kink}=(10i, 7, 4, 5.916),\qquad \mu=2+0.6i,\ \langle
a_{\rm breather}=(i,2,3i,2.45) .
\]
They intersect approximately at the point $(-0.930,0,0626)$. After the
collision, the position of kink gets shifted forward by $0.455$
and the position of breather gets shifted backward by $0.759$.

\section{Conclusion}
In this paper we have found and studied kink, breather and multi-soliton 
solutions of the vector sine-Gordon equation (\ref{vsG}). The soliton solutions 
obtained in this paper are of rank one and are generic. Actually 
a kink solution cannot be of rank more than one. For example, in order to
construct 
a kink solution of rank two, one needs to have two vectors   $$\langle 
a_1=(ip_1,q_1,{\bf a}_1^T), \ \langle a_2=(ip_2,q_2,{\bf a}_2^T),\ 
p_1,p_2,q_1,q_2\in\bbbr, \ {\bf a_1,a_2}\in \bbbr^n$$
satisfying the conditions 
$\langle a_1a_1\rangle=0,\  \langle a_2a_2\rangle=0,\ \langle a_1a_2\rangle=0$.
This implies that vectors $\langle a_1$ and $\langle a_2$ are linearly 
dependent.
However, in the case of breathers the corresponding vectors are 
complex and therefore there is no such obstacle. Thus higher rank breather 
solutions may exist.

Our results show that a $k$-kink solution of the $n$-component 
($\vec\alpha\in\bbbr^n$) vector sine-Gordon equation ($k\le n$)  is also 
a solution of $k$-component sine-Gordon equation. Similarly, a general   
$k$-breather solution of $n$-component equation (\ref{vsG}) ($k\le n/2)$ is
a solution of $2k$ component vector sine-Gordon equation. In particular, a 
general breather solution of the vector sine-Gordon equation cannot be obtained 
from any solution of the scalar sine-Gordon equation, but is a solution of a 
two-component vector sine-Gordon equation. 

The result of interactions 
(positional shifts and phase shifts) due to collisions in the $m_1$-kink and 
$m_2$-breather solution of the vector sine-Gordon equation (\ref{vsG}) is 
exactly the same as in the case of the scalar 
sine-Gordon equation (\ref{ssg}), which has been studied in \cite{fad87}. It 
only depends on the positions of poles in the dressing matrices.
This is significantly different from the case of the vector generalisation of 
nonlinear Schr{\"o}dinger equation, where the result of interactions depends 
on the initial polarisation of the colliding solitons \cite{manakov}. This can 
be explained by the difference in the spectral and symmetry properties of their 
Lax operators. These properties have been essentially used in the construction 
of multi-soliton solutions.

\section*{Acknowledgements}
The paper is supported by the Leverhulme Trust, AVM's EPSRC grant EP/I038675/1 and JPW's EPSRC grant EP/I038659/1. 
All authors gratefully acknowledge the financial support. JPW would like to thank S.C. Anco for providing some references.


\begin{thebibliography}{10}

\bibitem{PoR79}
K.~Pohlmeyer and K.‐H. Rehren.
\newblock Reduction of the two-dimensional ${O(n)}$ nonlinear $\sigma$-model.
\newblock {\em Journal of Mathematical Physics}, 20(12):2628--2632, 1979.

\bibitem{mr80c:81115}
V.~E. Zakharov and A.~V. Mikhailov.
\newblock Relativistically invariant two-dimensional models of field theory
  which are integrable by means of the inverse scattering problem method.
\newblock {\em Zh. {\`E}ksper. Teoret. Fiz.}, 74(6):1953--1973, 1978.

\bibitem{EP79}
H.~Eichenherr and K.~Pohlmeyer.
\newblock Lax pairs for certain generalizations of the sine-{G}ordon equation.
\newblock {\em Physics Letters B}, 89(1):76--78, 1979.

\bibitem{bps95}
Ioannis Bakas, Q-Han Park, and Hyun-Jong Shin.
\newblock Lagrangian formulation of symmetric space sine-{G}ordon models.
\newblock {\em Physics Letters B}, 372(1–2):45--52, 1996.

\bibitem{wang02}
Jing~Ping Wang.
\newblock Generalized {H}asimoto transformation and vector {S}ine-{G}ordon
  equation.
\newblock In S.~Abenda, G.~Gaeta, and S.~Walcher, editors, {\em Symmetry and
  Perturbation Theory, SPT 2002}. World Scientific, 2003.

\bibitem{MR2058803}
J.~A. Sanders and Jing~Ping Wang.
\newblock Integrable systems in \(n\)-dimensional {R}iemannian geometry.
\newblock {\em Moscow Mathematical Journal}, 3(4):1369--1393, 2003.

\bibitem{anco06}
S.C. Anco.
\newblock Hamiltonian flows of curves in symmetric spaces ${G/SO(N)}$ and
  vector soliton equations of m{K}d{V} and {S}ine-{G}ordon type.
\newblock {\em Symmetry, Integrability and Geometry: Methods and Applications},
  2:044, 2006.

\bibitem{ancow05}
S.~C. {Anco} and T.~{Wolf}.
\newblock Some symmetry classifications of hyperbolic vector evolution
  equations.
\newblock {\em Journal of Nonlinear Mathematical Physics}, 12(Supplement
  1):13--31, 2005.

\bibitem{II13}
Vladimir~G. Ivancevic and Tijana~T. Ivancevic.
\newblock Sine--{G}ordon solitons, kinks and breathers as physical models of
  nonlinear excitations in living cellular structures.
\newblock {\em Journal of Geometry and Symmetry in Physics}, 31:1--56, 2013.



\bibitem{manakov}
S.V. Manakov.
\newblock On the theory of two-dimensional stationary self-focusing of 
electromagnetic 
waves.
\newblock {\em Sov. Phys. JETP}, 38(2):248--253, 1974.

\bibitem{FK}
Allan P. Fordy and Peter P. Kulish.
\newblock Nonlinear {S}chr{\"o}dinger equations and simple {L}ie algebras.
\newblock {\em Communications in Mathematical Physics}, 89(3):427--443, 1983.

\bibitem{zash}
V.E. Zakharov and A.B. Shabat.
\newblock Integration of nonlinear equations of mathematical physics by the
  method of inverse scattering. {II}.
\newblock {\em Functional Analysis and Its Applications}, 13(3):166--174, 1979.

\bibitem{mik81}
A.V. Mikhailov.
\newblock The reduction problem and the inverse scattering method.
\newblock {\em Phys. D}, 3(1\& 2):73--117, 1981.

\bibitem{fad87}
L.D. Faddeev and L.A. Takhajan.
\newblock {\em Hamiltonian Methods in the Theory of Solitons}.
\newblock Springer Verlag, Berlin, 1987.

\bibitem{anco11}
Stephen~C. Anco, Nestor~Tchegoum Ngatat, and Mark Willoughby.
\newblock Interaction properties of complex modified {K}orteweg-de {V}ries
  (m{KdV}) solitons.
\newblock {\em Physica D: Nonlinear Phenomena}, 240(17):1378--1394, 2011.

\bibitem{mik79}
A.V. Mikhailov.
\newblock Integrability of a two-dimensional generalization of the {T}oda
  chain.
\newblock {\em JETP Lett.}, 30(7):414--418, 1979.

\bibitem{mik80}
A.V. Mikhailov.
\newblock Reduction in integrable systems. {T}he reduction group.
\newblock {\em JETP Lett.}, 32(2):187--192, 1980.

\end{thebibliography}

\begin{figure}[h!]
\begin{center}
 \includegraphics{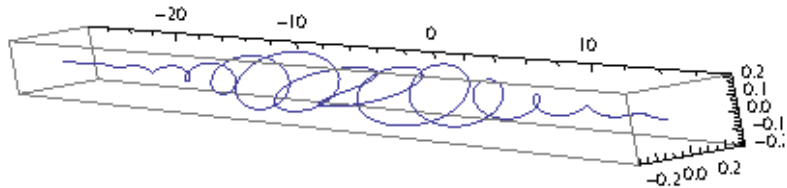}
 \end{center}\vspace{-0.7cm}
\caption{One breather solution. The curve  $(x,\alpha_1,\alpha_2)$ with 
$\mu=2+0.25i$ is plotted at time $t=3$.} \label{fig1}
\end{figure}

\begin{figure}[h!]
\begin{center}
 \includegraphics{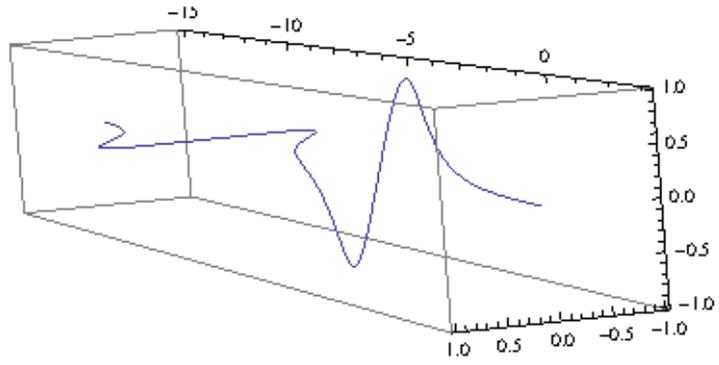}
\end{center}
\vspace{-1.1cm}
\caption{Two kink solution. The curve  $(x,\alpha_1,\alpha_2)$ with 
$\nu_1=.9,\ \nu_2=1.1$ is plotted at time $t=6$.}\label{fig2} 
\end{figure}

\begin{figure}[h!]
\begin{center}
\includegraphics[width=7cm]{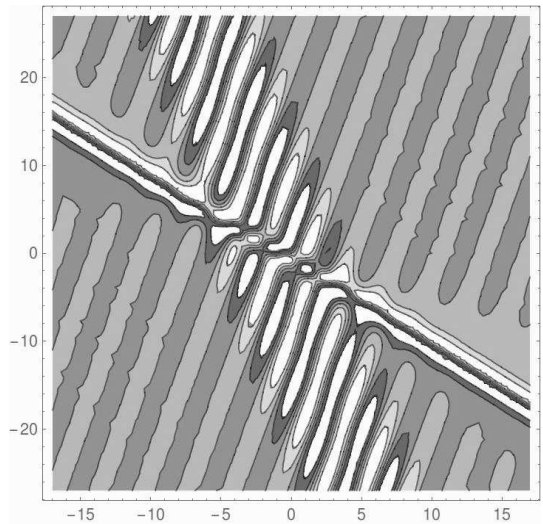}
\end{center}
\vspace{-0.3cm}
\caption{The contour plot of $\alpha_1$ in $(x,t)$-plane: Interaction of a kink
with a
breather}\label{fig3}
\end{figure}

\end{document}